\documentclass[superscriptaddress,aps,letterpaper,10pt,twocolumn,floats,amsmath,amsfonts,amssymb,prl]{revtex4-2}

\usepackage{graphicx}
\usepackage[colorlinks=true,citecolor=blue]{hyperref}
\usepackage[caption=false]{subfig}
\usepackage{pstricks}
\usepackage{pst-node}
\usepackage{amsmath}
\usepackage{amsbsy}
\usepackage{verbatim}
\usepackage{float}
\usepackage{amsthm}
\newtheorem{theorem}{Theorem}
\newtheorem{definition}{Definition}

\begin{document}
	
	
	\title{Fermionic one-body entanglement as a thermodynamic resource}
	
\author{Krzysztof Ptaszy\'{n}ski}
\email{krzysztof.ptaszynski@uni.lu}
\affiliation{Complex Systems and Statistical Mechanics, Department of Physics and Materials Science, University of Luxembourg, L-1511 Luxembourg, Luxembourg}
\affiliation{Institute of Molecular Physics, Polish Academy of Sciences, Mariana Smoluchowskiego 17, 60-179 Pozna\'{n}, Poland}

\author{Massimiliano Esposito}
\email{massimiliano.esposito@uni.lu}
\affiliation{Complex Systems and Statistical Mechanics, Department of Physics and Materials Science, University of Luxembourg, L-1511 Luxembourg, Luxembourg}
	
	\date{\today}
	
	\begin{abstract}
		There is ongoing controversy about whether a coherent superposition of the occupied states of two fermionic modes should be regarded entangled or not, that is, whether its intrinsic quantum correlations are operationally accessible and useful as a resource. This has been questioned on the basis that such an entanglement cannot be accessed by local operations on individual modes due to the parity superselection rule which constrains the set of physical observables. In other words, one cannot observe violations of Bell's inequality. Here we show, however, that entanglement of a two-mode fermionic state can be used as a genuine quantum resource in open-system thermodynamic processes, enabling one to perform tasks forbidden for separable states. We thus demonstrate that quantum thermodynamics can shed light on the nature of fermionic entanglement and the operational meaning of the different notions used to define it.
	\end{abstract}
	
	\maketitle
	
	\textit{Introduction}.---Entanglement is one of the most fundamental concepts of quantum mechanics, implying the impossibility of local realistic description of physical systems, as well as one of key resources in quantum computing, metrology, communication, and thermodynamics~\cite{horodecki2009}. In the field of quantum information, it is commonly defined as follows~\cite{nielsen2010}:
	\begin{definition}[General definition of entanglement] \label{statebased}
		The state of a bipartite system $AB$ is deemed separable when its density matrix can be decomposed as
		\begin{align}
			\rho_{AB} =\sum_{k} \lambda_k \rho_{A}^k \otimes \rho_{B}^k,
		\end{align}
		where $\lambda_k$ are positive-valued probabilities summing up to 1 and $\rho_{A}^k$, $\rho_{B}^k$ are positive semidefinite matrices with trace 1. Otherwise, the state is deemed to be entangled.
	\end{definition}
	Although this general notion is applicable to systems of distinguishable entities (such as qubits), the proper definition of entanglement of indistinguishable particles is still a disputed issue (see Ref.~\cite{benatti2020} for a review). For example, it has been asserted that the definition of entanglement needs to be modified for fermionic systems which obey the parity superselection rule prohibiting coherent superpositions of states with different particle parity (e.g., empty and occupied states of a single energy level)~\cite{wick1952, szalay2021}. This rule provides constraints on the form of physically allowed fermionic states, as well as observables and quantum operations~\cite{vidal2021}. As thoroughly discussed by Ba\~{n}uls \textit{et al.}~\cite{banuls2007}, applying the parity superselection rule in different ways, one can provide physically justified notions of entanglement which are either weaker~\cite{dariano2014,dariano2014b} or stronger~\cite{wiseman2003, ding2021, ding2022, ernst2022, wiseman2004} than Definition~\ref{statebased}. In particular, these definitions provide different answers to the question whether the pure state of a single fermion delocalized between two modes $A$ and $B$ (e.g., quantum dots),
	\begin{align}
		|\Psi_{AB} \rangle=\frac{1}{\sqrt{2}} \left(\left|1_A \right \rangle \left|0_B \right \rangle+ \left|0_A \right \rangle \left|1_B \right \rangle \right),
	\end{align}
	is entangled or not (here $|0_X \rangle$ and $|1_X \rangle$ denote the empty and occupied states of the mode $X$, respectively; see footnote~\cite{commentnot} for further clarifications about the notation used). According to Definition~\ref{statebased} such a state is maximally entangled [indeed, it looks like the Bell state ${(|\mathord{\uparrow} \mathord{\downarrow} \rangle+|\mathord{\downarrow}\mathord{\uparrow} \rangle)/\sqrt{2}}$]. It is separable, however, according to a more restrictive notion used in Refs.~\cite{wiseman2003, ding2021, ding2022, ernst2022, wiseman2004}:
	\begin{definition}[Observable-based definition of fermionic entanglement] \label{obsbased}
		Let us first define the locally projected state
		\begin{align}
			\pi_{AB}=\sum_{\alpha,\gamma=e,o} \left( \mathbb{P}_\alpha^A \otimes \mathbb{P}_\gamma^B \right) \rho_{AB} \left( \mathbb{P}_\alpha^A \otimes \mathbb{P}_\gamma^B \right),
		\end{align}
		where $\mathbb{P}_{e/o}^X$ is the local projection of the subsystem $X \in \{A,B\}$ on the even/odd particle parity sector. Then, the state $\rho_{AB}$ is considered entangled when $\pi_{AB}$ is entangled with respect to Definition~\ref{statebased}.
	\end{definition}
	The physical meaning of this definition becomes clear by noting that the state $\pi_{AB}$ reproduces all the correlations of local observables $O_A$ and $O_B$ acting on $A$ and $B$,
	\begin{align}
		\forall O_A, O_B: \quad \text{Tr} \left[O_A O_B \rho_{AB} \right] = \text{Tr} \left[O_A O_B \pi_{AB} \right],
	\end{align}
	as the observables obey the parity superselection rule. Therefore, the state $\rho_{AB}$ is deemed separable when it cannot be distinguished from a classical mixture of tensor product states via correlations of local measurements (e.g., through a violation of Bell's inequality). The reason why the state $|\Psi_{AB} \rangle$ is not entangled according to Definition~\ref{obsbased} is that the parity superselection rule permits only a single type of measurement on a single fermionic mode, namely, a projection on its empty or occupied state. As a result, one cannot distinguish the state $|\Psi_{AB} \rangle$ from the mixed state
	\begin{align}
		\pi_{AB} = \frac{1}{2} \left( \left|1_A \right \rangle \left|0_B \right \rangle \left \langle 1_A \right| \left \langle 0_B \right| + \left|0_A \right \rangle \left|1_B \right \rangle \left \langle 0_A \right| \left \langle 1_B \right| \right),
	\end{align}
	and thus this state is deemed separable. In fact, any pure or mixed state $\rho_{AB}$ of a two-mode system is separable according to Definition~\ref{obsbased}. This is the reason why the operational meaning of the entanglement of the one-body state with respect to Definition~\ref{statebased} has been put into question. It has been described as a physically inaccessible ``fluffy bunny'' entanglement~\cite{wiseman2004, ernst2022} and its applicability as a quantum resource has been contested~\cite{ding2021}.
	
At the same time, it has been shown that the violation of Bell's inequality~\cite{dasenbrook2016} or quantum teleportation~\cite{olofsson2020, debarba2020} can be achieved using an initial product state of a pair of one-body states: $|\Psi_{\tilde{A} \tilde{B}} \rangle=|\Psi_{A_1B_1} \rangle \otimes |\Psi_{A_2B_2} \rangle$. This has been interpreted as a demonstration of a genuine one-body entanglement~\cite{friis2016}. However, one can raise a possible objection to such an interpretation, which may read as follows. The aforementioned protocols require the joint and simultaneous processing of the states $|\Psi_{A_1B_1} \rangle$ and $|\Psi_{A_2B_2} \rangle$, and the measurement of their collective observables [defined as superpositions of different states of modes $A_i$, such as ${(|1_{A_1} \rangle |0_{A_2} \rangle+|0_{A_1} \rangle |1_{A_2} \rangle)/\sqrt{2}}$]. This differs from the standard experimental procedures used to demonstrate the presence of entanglement, where each entangled state is processed in a separate experimental run (e.g., individual photon pairs are created and detected at separate moments of time). Thus, one may argue that the aforementioned studies can be interpreted on the ground of Definition~\ref{obsbased} using the concept of entanglement \textit{super-additivity}~\cite{wiseman2003}: Due to the possibility of defining collective observables, the tensor product $|\Psi_{\tilde{A} \tilde{B}} \rangle=|\Psi_{A_1B_1} \rangle \otimes |\Psi_{A_2B_2} \rangle$ can be entangled even when the individual states $|\Psi_{A_iB_i} \rangle$ are separable. According to this line of reasoning, the entanglement of a single one-body state $|\Psi_{AB} \rangle$ is regarded physically meaningful only if it is operationally accessible when such a state (or, more generally, every member of an ensemble of such states) is processed separately (as in typical Bell's inequality violation experiments).
	
In this Letter, we show that considering the thermodynamics of the open system provides an operational way to physically reveal the entanglement of a single one-body fermionic state $|\Psi_{AB} \rangle$ with respect to Definition~\ref{statebased}. Such a state can be used to realize thermodynamic tasks that would be impossible with a separable state. Specifically, we demonstrate a process in which the one-body entanglement is a source of positive coherent information (negative conditional entropy) that can be used as a thermodynamic resource to perform the cooling~\cite{rio2011}. This shows that the entanglement of the one-body state has a clear physical manifestation and operational value, even though it is inaccessible through local operations.
	
	\textit{Setup}.---We consider a single-mode system $S$ coupled with a single-mode memory $M$ described by the Hamiltonian
	\begin{align}
		H_{MS}=\epsilon_M c_M^\dagger c_M + \epsilon_S(t) c_S^\dagger c_S + \hbar \Omega(t) \left(c_M^\dagger c_S+c_S^\dagger c_M \right),
	\end{align}
	where $\epsilon_X$ is the energy of the mode $X$ and $\Omega(t)$ is the tunnel coupling. The system-memory ensemble is connected to the reservoir $R$ as
	\begin{align}
		H_{MSR}=H_{MS}(t)+H_R+H_I(t),
	\end{align}
	where $H_R$ is the reservoir Hamiltonian and $H_I(t)$ describes a switchable system-reservoir coupling. The global system $MSR$ undergoes a unitary evolution starting from the initial state $\rho_{MSR}(0)=\rho_{MS}(0) \otimes \rho_R^\text{eq}$, where $\rho_R^\text{eq}=\exp(-\beta H_R)/\text{Tr} \exp(-\beta H_R)$ is the equilibrium Gibbs state of the reservoir. Here $\beta=1/(k_B T)$ is the inverse temperature of the reservoir, and we take the chemical potential $\mu=0$ for simplicity. The heat flow is defined as an energy change of the reservoir (with a minus sign):
	\begin{align}
		-Q = \text{Tr} \left\{ H_R \left[ \rho_R(t) - \rho_R^\text{eq} \right] \right\}.
	\end{align}
	Here, we follow a standard sign convention where the heat flow to the system is denoted as positive. 
	
	Using such definitions, the heat flow has been shown to obey the second law of thermodynamics in the form~\cite{esposito2010}
	\begin{align}
		\sigma=\Delta S_{MS}-\beta Q \geq 0,
	\end{align}
	where $\sigma$ is the entropy production, $S_X=-\text{Tr} (\rho_X \ln \rho_X)$ is the von Neumann entropy, and $\Delta S_{MS}=S_{MS}(t)-S_{MS}(0)$; this expression is valid for an arbitrary initial state $\rho_{MS}$ (since the von Neumann entropy is well defined outside of equilibrium).
	
	\textit{Memory-assisted purification}.---We will now consider a thermodynamic process proposed by del Rio \textit{et al.}~\cite{rio2011}, which is defined as follows:
	\begin{definition}[Memory-assisted purification] \label{memerasure} Memory-assisted purification is a process converting the initially-correlated state $\rho_{MS}$ of the memory $M$ and the purified system $S$ as
		\begin{align}
			\rho_{MS} \rightarrow \rho_M \otimes |\Psi_{S'} \rangle \langle \Psi_{S'}|,
		\end{align}
		where $\rho_M$ is the unchanged state of the memory and $|\Psi_{S'} \rangle$ is the final pure state of the system.
	\end{definition}
	
	Heat dissipation during this process obeys the following theorem~\cite{rio2011}:
	\begin{theorem} \label{distheorem}
		For any initial state $\rho_{MS}$ which is separable with respect to Definition~\ref{statebased} the heat dissipation during the memory-assisted purification is nonnegative: $-Q \geq 0$.
	\end{theorem}
	\begin{proof}
		The change in von Neumann entropy of $\rho_{MS}$ during memory-assisted purification $\Delta S_{MS}=S_M-S_{SM}$ is equal to the coherent information $I_{S \rangle M}=S_{M}-S_{SM}$, which is negative for any separable state~\cite{cerf1997, vollbrecht2002} (with respect to Definition~\ref{statebased}). According to the second law of thermodynamics $\Delta S_{MS}-\beta Q \geq 0$. Thus, $-Q \geq 0$.
	\end{proof}
	Therefore, the cooling of the reservoir ($-Q<0$) during memory-assisted purification can be realized only when the initial state $\rho_{MS}$ is entangled. This can be also expressed in another way: Entanglement allows one to extract work ($-W =Q>0$) during the process since it increases the initial free energy $F_{MS}=E_{MS}-\beta^{-1} S_{MS}$ (which obeys the inequality $-W \leq -\Delta F_{MS}$) by lowering the entropy.
	
	We will now demonstrate that the reservoir cooling can be achieved using physically-allowed fermionic operations and the initial one-body state
	\begin{align}
		|\Psi_{MS} \rangle^{(1)}=\frac{1}{\sqrt{2}} \left(\left|1_M \right \rangle \left|0_S \right \rangle+i \left|0_M \right \rangle \left|1_S \right \rangle \right).
	\end{align}
	The initial reduced states of the system and the memory read accordingly:
	\begin{align}
		\rho_S^{(1)} &=\frac{1}{2} \left( \left|1_S \right \rangle \left \langle 1_S \right| + \left|0_S \right \rangle \left \langle 0_S \right| \right), \\ \label{memred}
		\rho_M^{(1)} &=\frac{1}{2} \left( \left|1_M \right \rangle \left \langle 1_M \right| + \left|0_M \right \rangle \left \langle 0_M \right| \right).
	\end{align}
	The von Neumann entropy of the total pure state $|\Psi_{MS} \rangle^{(1)}$ is equal to 0 ($S_{MS}=0$). The marginal entropies of $S$ and $M$ read $S_S=S_M=h(1/2)=\ln 2$ where $h(x)=-x \ln x-(1-x)\ln(1-x)$ is the binary entropy. Thus, the coherent information $I_{S \rangle M}=S_M-S_{SM}=\ln 2$ is positive.
	
	We shall now consider the following sequence of operations. The initial mode energies are set to 0: $\epsilon_M=\epsilon_S=0$. In the first step, the initial state is converted via a unitary transformation as
	\begin{align} \label{unitconv1}
		|\Psi_{MS} \rangle^{(1)} \rightarrow |\Psi_{MS} \rangle^{(2)} = \left|0_M \right \rangle \left|1_S \right \rangle,
	\end{align}
	which is performed by switching on the tunnel coupling $\Omega(t)=\Omega$ for a time $\pi \Omega^{-1}/4$. This can be realized experimentally, e.g., in quantum dot systems~\cite{hayashi2003, gorman2005}. We note here that the coherent superposition $|\Psi_{MS} \rangle^{(1)}$ can be prepared by the reverse transformation from the initial state $\left|0_M \right \rangle \left|1_S \right \rangle$ which, in turn, can be created by nonequilibrium driving~\cite{hayashi2003} or by Landauer erasure~\cite{scandi2022}.
	
	In the next step, we perform the following thermodynamic quasistatic process: The system energy is lowered to $\epsilon_S \rightarrow -\infty$. Then, the system is coupled to the reservoir and its energy is slowly increased to $\epsilon_S=0$ (see Ref.~\cite{koski2014} for a similar experimental procedure). As a result, the population of the system decreases from 1 to 1/2 and the state becomes mixed:
	\begin{align} \label{step2}
		&|\Psi_{MS} \rangle^{(2)} \rightarrow \rho_{MS}^{(3)} \\ \nonumber &=\frac{1}{2} \left|0_M \right \rangle \left \langle 0_M \right| \otimes  \left( \left|1_S \right \rangle \left \langle 1_S \right| + \left|0_S \right \rangle \left \langle 0_S \right| \right).
	\end{align}
	During a quasi-static (i.e., reversible) process, the entropy production $\sigma=\Delta S_S-\beta Q$ is equal to 0. Therefore, the heat dissipation is negative:
	\begin{align} \nonumber
		-Q =& -k_B T \Delta S_S=-k_B T [h(1/2)-h(1)]= \\ &-k_B T \ln 2=-k_B T I_{S \rangle M}<0.
	\end{align}
	In this way, the coherent information $I_{S \rangle M}$ is fully used as a resource to perform the cooling. Finally, we apply a unitary process swapping the states of the system and the memory,
	\begin{align} \label{step3}
		&\rho_{MS}^{(3)} \rightarrow \rho_{MS}^{(4)}=\\ \nonumber
		&\frac{1}{2} \left( \left|1_M \right \rangle \left \langle 1_M \right| + \left|0_M \right \rangle \left \langle 0_M \right| \right) \otimes \left|0_S \right \rangle \left \langle 0_S \right|,
	\end{align}
	which is realized by turning on the tunnel coupling $\Omega(t)=\Omega$ for a time $\pi \Omega^{-1}/2$. As a result, the system becomes purified while the memory state is restored to the initial one [cf.~Eq.~\eqref{memred}]. At this point, the memory-assisted purification is complete.
	
	\textit{Validity of quasistatic approximation}.---While our analytic results assume an ideal quasi-static process, we now show that cooling can also be achieved in finite time using finite level energies. Specifically, we consider a process where the system energy is linearly swept from the value $\epsilon_1<0$ to $\epsilon_2>0$ during interval $[0,\tau]$, and then stays at the value $\epsilon_2$. The dynamics of the system population is described by the master equation~\cite{fichetti1998}
	\begin{align}
		\dot{n}_S=-\Gamma \left[n_S-f(\epsilon_S) \right],
	\end{align}
	where $\Gamma$ is the relaxation rate and $f(\epsilon_S)={[1+\exp(\beta \epsilon_S)]^{-1}}$ is the Fermi distribution. For weak system-reservoir coupling $\Gamma \ll k_B T$ the heat dissipation can be calculated as~\cite{benenti2017}
	\begin{align}
		-Q=-\int_{0}^{t_f} \epsilon_S(t) \dot{n}_S(t) dt,
	\end{align}
	where $t_f$ is the time at which $n_S(t)$ reaches 1/2 and the coupling to the bath is switched-off. We will later show that this formula well reproduces the heat calculated using exact simulations of the system-reservoir dynamics which takes into account finite $\Gamma$. The results are presented in Fig.~\ref{fig:heatmark}. As shown, heat dissipation becomes negative for $\Gamma \tau \gtrapprox 4$ and approaches a theoretical limit $-k_B T \ln 2$ for long sweep times.
	%
\begin{figure}
	\centering
	\includegraphics[width=0.93\linewidth]{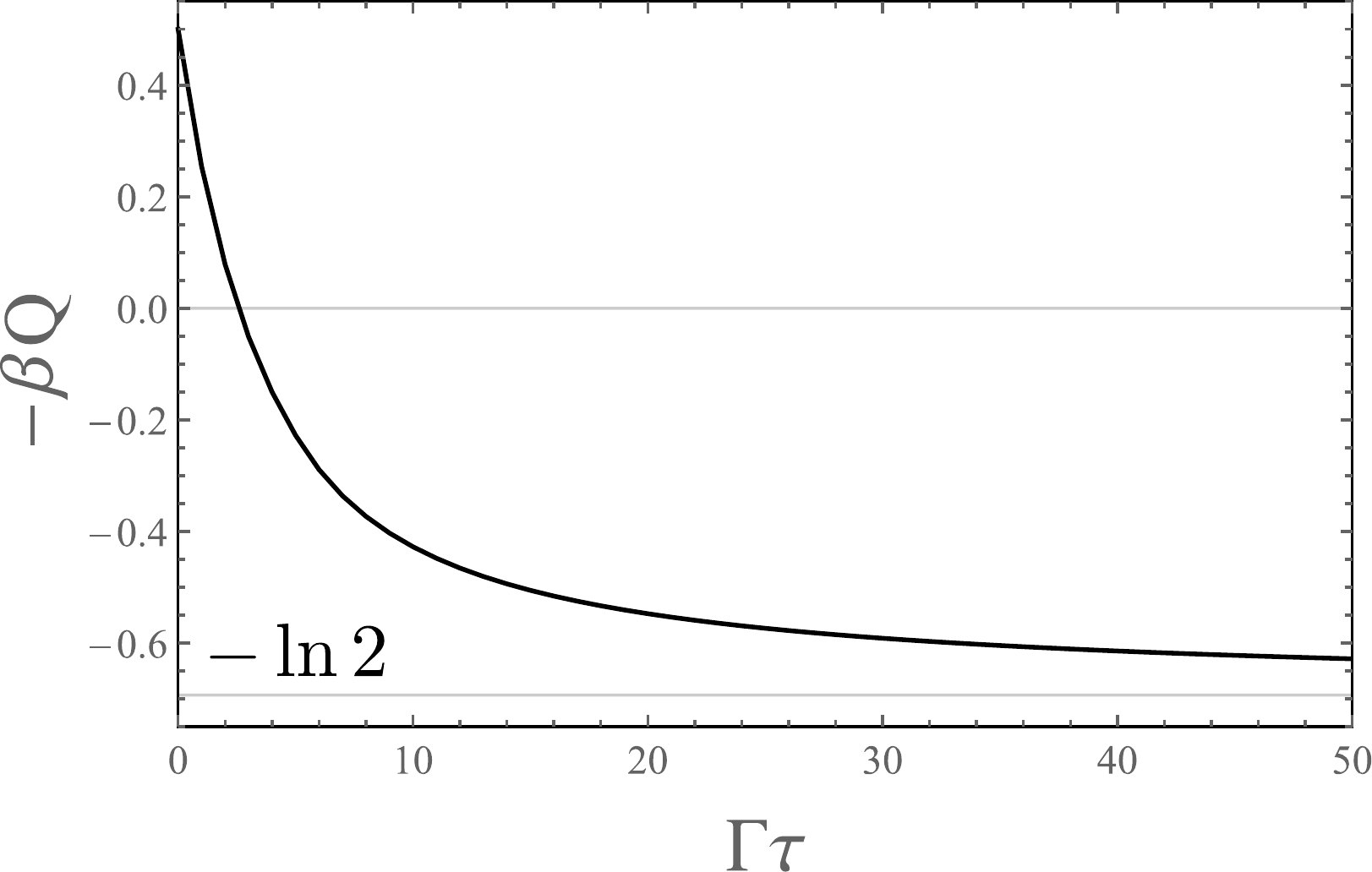}		
	\caption{Heat dissipation as a function of sweep time $\tau$ for $\epsilon_1=-5 k_B T$ and $\epsilon_2=k_B T$.}
	\label{fig:heatmark}
\end{figure}
%
	
	To confirm the validity of the master equation approach, we performed a simulation of the exact system-reservoir dynamics generated by the Hamiltonian
	\begin{align}
		H_{SR}=\epsilon_S c^\dagger_S c_S + \sum_{k=1}^K \epsilon_k c_k^\dagger c_k + \mathcal{T} \sum_{k=1}^K \left(c_k^\dagger c_S+c_S^\dagger c_k \right).
	\end{align}
	The energy levels of the reservoir $\epsilon_k$ are uniformly distributed over the interval $[-7k_B T,3k_B T]$. $\mathcal{T}$ is related to $\Gamma$ as $\Gamma=2\pi \mathcal{T}^2 \xi$, where $\xi=K/(10k_B T)$ is the density of states in the reservoir~\cite{schaller2014}. Since the Hamiltonian is quadratic, the system-reservoir state can be fully characterized by the correlation matrix $\mathcal{C}_{ij}=\langle c_i^\dagger c_j \rangle$~\cite{peschel2003}. We simulated its evolution using the Euler method~\cite{eisler2012}:
	\begin{align}
		\mathcal{C}(t+\Delta t)=e^{i \Delta t \mathcal{H}(t)} \mathcal{C}(t) e^{-i \Delta t \mathcal{H}(t)}.
	\end{align}
	%
	\begin{figure}
		\centering
		\includegraphics[width=0.93\linewidth]{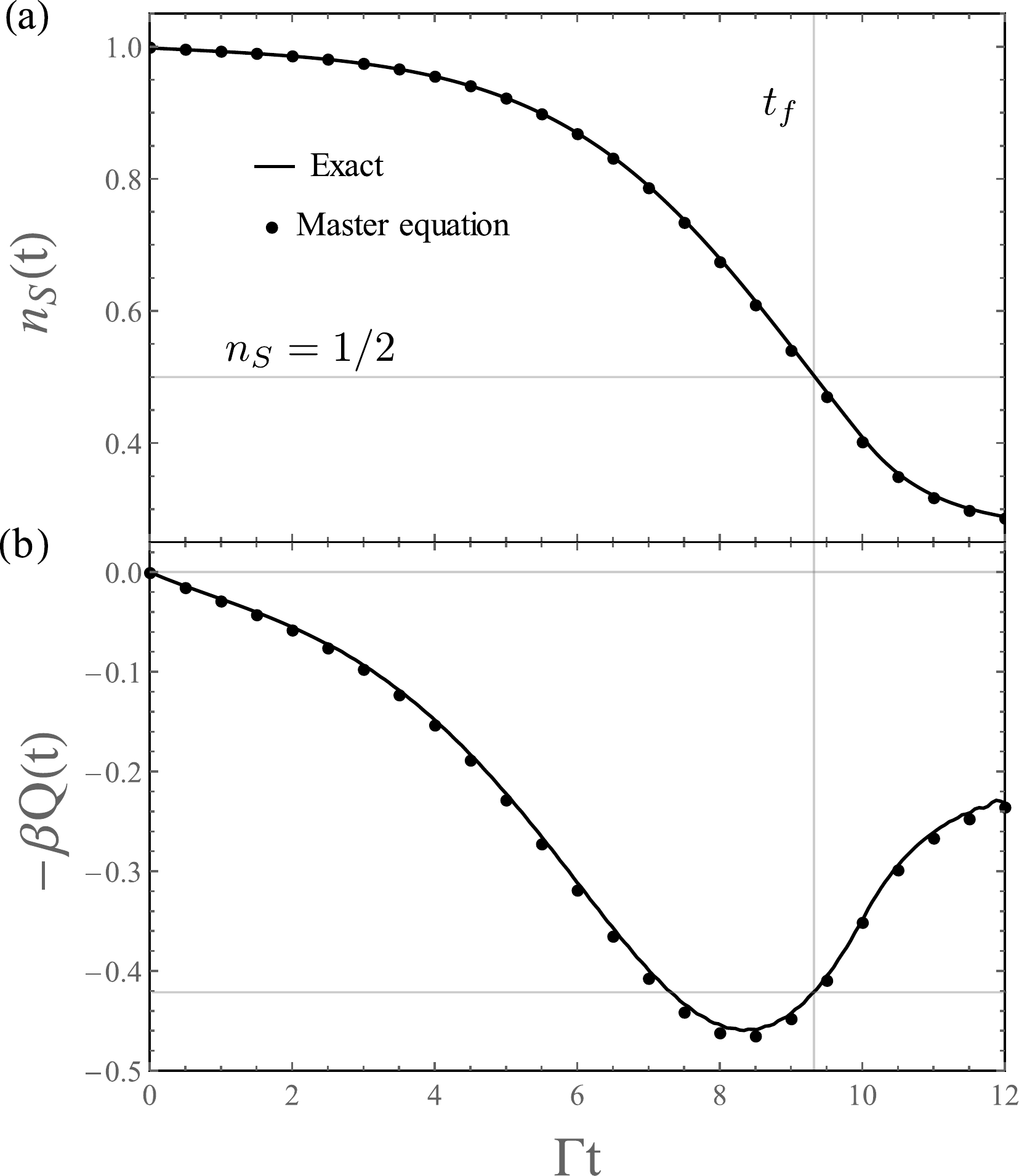}		
		\caption{(a) State population and (b) heat dissipation as a function of time for the exact dynamics (black solid line) and master equation (black dots). Parameters: $\Gamma=0.02 k_B T$, $\Gamma \tau=10$, $\Gamma \Delta t=0.06$, $K=200$.}
		\label{fig:comp}
	\end{figure}
	%
	As a matter of fact, this is not just an approximation, but corresponds to a physical scenario when $\epsilon_S$ is swept in a stepwise manner. The matrix $\mathcal{H}$ is defined as $\mathcal{H}_{ii}=\epsilon_i$, $\mathcal{H}_{Sk}=\mathcal{H}_{kS}=\mathcal{T}$, and $\mathcal{H}_{ij}=0$ otherwise. The initial state reads $\mathcal{C}(0)=\text{diag}[1,f(\epsilon_1),\ldots,f(\epsilon_K)]$. The reservoir energy can be calculated as $\text{Tr} (H_R \rho_R)=\sum_{k=1}^K \epsilon_k c_{kk}$. The initial and final energies of the system were chosen as $\epsilon_1=-5 k_B T$ and $\epsilon_2=k_BT$. An example of the time evolution of the state population and heat dissipation is presented in Fig.~\ref{fig:comp}. One notes that the exact results are in perfect agreement with the predictions of the master equation and also describe cooling. Indeed, heat dissipation is negative $(-Q \approx -0.42 k_B T <0)$ for $\Gamma t =\Gamma t_f \approx 9.3$ when $n_S(t)$ reaches the value 1/2. 
	
	\textit{Entanglement detection}.---Finally, let us show how the presence of one-body entanglement can be experimentally detected (also for processes which do not fulfill Definition~\ref{memerasure}). This can be realized as follows. One prepares an ensemble of two-mode states $\rho_{MS}$. For some members of the ensemble, one measures the initial occupancies $n_S^0$ and $n_M^0$. For the remaining members, one performs the sequence of quantum operations and determines the final occupancies $n_S^1$, $m_M^1$, as well as heat dissipation $-Q$. The minimum initial entropy for separable states is $\text{max}[h(n_S^0),h(n_M^0)]$~\cite{cerf1997, vollbrecht2002}. Thus, the presence of an initial entanglement is implied whenever the inequality
	\begin{align} \label{nonideq}
		h(n_S^1)+h(n_M^1)-\text{max}\left[h(n_S^0),h(n_M^0) \right]-\beta Q \geq 0
	\end{align}
	is violated (the proof is analogous to Theorem~\ref{distheorem}). In fact, this inequality is already broken for the unitary process given by Eq.~\eqref{unitconv1}, where the system is thermally isolated from the environment ($-Q=0$). Since coherent electron dynamics in isolated quantum dots has been experimentally realized~\cite{gorman2005}, the detection of the one-body fermionic entanglement appears to be feasible using state-of-the-art techniques. We highlight that in the proposed protocol, each state $\rho_{MS}$ can be processed separately. For example, one may use identically prepared states of the same double-quantum-dot system, which are processed sequentially at separate moments of time. Therefore, it can be used to demonstrate a genuine one-body entanglement with respect to Definition~\ref{statebased} (cf.~the discussion in the Introduction).
	
	\textit{Final remarks}.---We demonstrate that entanglement with respect to Definition~\ref{statebased} retains its operational meaning also for two-mode fermionic states obeying the parity superselection rule (which has been questioned in Refs.~\cite{wiseman2003, wiseman2004, ding2021, ding2022, ernst2022}): It describes the potential of a quantum state to carry the positive coherent information (negative conditional entropy), which can be extracted via global operations (e.g., open-system thermodynamic processes) to perform classically forbidden tasks. Since positive coherent information is a resource in several quantum information protocols~\cite{bruss2004, prabhu2013, yang2019, vempati2021}, the applicability of fermionic states to these goals needs to be properly established. At the same time, such an entanglement is not enough to observe non-local correlations between distant parties, which requires the fulfillment of Definition~\ref{obsbased}; however, this obstacle can be easily circumvented by jointly processing a pair of fermionic states~\cite{dasenbrook2016, olofsson2020, debarba2020}. Another (and equally justified) interpretation of our result is that the entanglement of a one-body state is a genuine quantum resource, which can be extracted in two distinct ways: either by applying only local operations, but allowing joint processing of several copies of a quantum state (as in Refs.~\cite{dasenbrook2016, olofsson2020, debarba2020}), or by separately processing different copies, but allowing global operations (our proposal). Furthermore, we note that fermionic entanglement defined using notions even weaker than Definition~\ref{statebased}~\cite{dariano2014, dariano2014b} has been demonstrated to be applicable for certain tasks, such as data hiding~\cite{verstraete2003, schuch2004} (due to the impossibility of local state preparation). We conclude, therefore, that different complementary definitions of fermionic entanglement -- and, more generally, entanglement between identical particles -- may be justified in different physical contexts (see Ref.~\cite{morris2020} for a similar message concerning the particle entanglement in bosonic systems). Our Letter shows that quantum thermodynamics may be a useful theoretical framework to shed light on this issue.
	
	\begin{acknowledgments}
	K. P. is supported by the Scholarships of Minister of Science and Higher Education. This research was also supported by the FQXi foundation project FQXi-IAF19-05.
	\end{acknowledgments}


\begin{thebibliography}{}
		\bibitem{horodecki2009}
		R. Horodecki, P. Horodecki, M. Horodecki, and K. Horodecki, Quantum entanglement, \href{\doibase 10.1103/RevModPhys.81.865}{Rev. Mod. Phys. \textbf{81}, 865 (2009)}.
		
		\bibitem{nielsen2010}
		M. A. Nielsen and I. L. Chuang, \textit{Quantum Computation and Quantum Information} (Cambridge University Press, Cambridge, 2010).
		
		\bibitem{benatti2020}
		F. Benatti, R. Floreanini, F. Franchini, and U. Marzolino, Entanglement in indistinguishable particle systems, \href{\doibase 10.1016/j.physrep.2020.07.003}{Phys. Rep. \textbf{878}, 1 (2020)}.
		
		\bibitem{wick1952}
		G. C. Wick, A. S. Wightman, and E. P. Wigner, The Intrinsic Parity of Elementary Particles, \href{\doibase 10.1103/PhysRev.88.101}{Phys. Rev. \textbf{88}, 101 (1952)}.
		
		\bibitem{szalay2021}
		S. Szalay, Z Zimbor\'{a}s, M. M\'{a}t\'{e}, G. Barcza, C. Schilling, and \"{O}. Legeza, Fermionic systems for quantum information people, \href{\doibase 10.1088/1751-8121/ac0646}{J. Phys. A: Math. Theor. \textbf{54}, 393001 (2021)}.
		
		\bibitem{vidal2021}
		N. T. Vidal, M. L. Bera, A. Riera, M. Lewenstein, and M. N. Bera, Quantum operations in an information theory for fermions, \href{\doibase 10.1103/PhysRevA.104.032411}{Phys. Rev. A \textbf{104}, 032411 (2021)}.
		
		\bibitem{banuls2007}
		M.-C. Ba\~{n}uls, J. I. Cirac, and M. M. Wolf, Entanglement in fermionic systems, \href{\doibase 10.1103/PhysRevA.76.022311}{Phys. Rev. A \textbf{76}, 022311 (2007)}.
		
		\bibitem{dariano2014}
		G. M. D'Ariano, F. Manessi, P. Perinotti, and A. Tosini, Fermionic computation is non-local tomographic and violates monogamy of entanglement, \href{\doibase 10.1209/0295-5075/107/20009}{Europhys. Lett. \textbf{107}, 20009 (2014)}.
		
		\bibitem{dariano2014b}
		G. M. D'Ariano, F. Manessi, P. Perinotti, and A. Tosini, The Feynman problem and fermionic entanglement: Fermionic theory versus qubit theory, \href{\doibase 10.1142/S0217751X14300257}{Int. J. Mod. Phys. A \textbf{29}, 1430025 (2014)}.
		
		\bibitem{wiseman2003}
		H. M. Wiseman and J. A. Vaccaro, Entanglement of Indistinguishable Particles Shared between Two Parties, \href{\doibase 10.1103/PhysRevLett.91.097902}{Phys. Rev. Lett. \textbf{91}, 097902 (2003)}.
		
		\bibitem{wiseman2004}
		H. M. Wiseman, S. D. Bartlett, and J. A. Vaccaro, Ferreting Out the Fluffy Bunnies: Entanglement Constrained by Generalized Superselection Rules, \href{\doibase 10.1142/9789812703002_0047}{in \textit{Laser Spectroscopy}, edited by  P. Hannaford, A. Sidorov, H. Bachor, and K. Baldwin (World Scientific, Singapore, 2004), pp. 307–314}.
		
		\bibitem{ding2021}
		L. Ding, S. Mardazad, S. Das, S. Szalay, U. Schollw\"{o}ck, Z. Zimbor\'{a}s, and C. Schilling, Concept of Orbital Entanglement and Correlation in Quantum Chemistry, \href{\doibase 10.1021/acs.jctc.0c00559}{J. Chem. Theory Comput. \textbf{17}, 79 (2021)}.
		
		\bibitem{ding2022}
		L. Ding, Z. Zimbor\'{a}s, and C. Schilling, Quantifying Electron Entanglement Faithfully, \href{\doibase 10.48550/arXiv.2207.03377}{arXiv:2207.03377 (2022)}.
		
		\bibitem{ernst2022}
		J. O. Ernst and F. Tennie, Mode Entanglement in Fermionic and Bosonic Harmonium, \href{\doibase 10.48550/arXiv.2211.09647}{arXiv:2211.09647 (2022)}.

  \bibitem{commentnot}
  Here, for the sake of readability, we use a notation commonly used in quantum information to describe a system of two qubits; this is justified by the possibility of mapping a two-mode fermionic system onto a two-qubit system via Jordan-Wigner transformation. Accordingly, within this qubit-like notation, we use the standard tensor product symbol $\otimes$ to express the uncorrelated states of two subsystems.
		
		\bibitem{dasenbrook2016}
		D. Dasenbrook, J. Bowles, J. B. Brask, P. P. Hofer, C. Flindt, and N. Brunner, Single-electron entanglement and nonlocality, \href{\doibase 10.1088/1367-2630/18/4/043036}{New J. Phys. \textbf{18}, 043036 (2016)}.
		
		\bibitem{olofsson2020}
		E. Olofsson, P. Samuelsson, N. Brunner, and P. P. Potts, Quantum teleportation of single-electron states, \href{\doibase 10.1103/PhysRevB.101.195403}{Phys. Rev. B \textbf{101}, 195403 (2020)}.
		
		\bibitem{debarba2020}
		T. Debarba, F. Iemini, G. Giedke, and N. Friis, Teleporting quantum information encoded in fermionic modes, \href{\doibase 10.1103/PhysRevA.101.052326}{Phys. Rev. A \textbf{101}, 052326 (2020)}.
		
		\bibitem{friis2016}
		N. Friis, Unlocking fermionic mode entanglement, \href{\doibase 10.1088/1367-2630/18/6/061001}{New J. Phys. \textbf{18}, 061001 (2016)}.
		
		\bibitem{rio2011}
		L. del Rio, J. Åberg, R. Renner, O. Dahlsten, and V. Vedral, The thermodynamic meaning of negative entropy, \href{\doibase 10.1038/nature10123}{Nature (London) \textbf{474}, 61 (2011)}.
		
		\bibitem{esposito2010}
		M. Esposito, K. Lindenberg, and C. Van den Broeck, Entropy production as correlation between system and reservoir, \href{\doibase 10.1088/1367-2630/12/1/013013}{New J. Phys. \textbf{12}, 013013 (2010)}.
		
		\bibitem{cerf1997}
		N. J. Cerf and C. Adami, Negative Entropy and Information in Quantum Mechanics, \href{\doibase 10.1103/PhysRevLett.79.5194}{Phys. Rev. Lett. \textbf{79}, 5194 (1997)}.
		
		\bibitem{vollbrecht2002}
		K. G. H. Vollbrecht and M. M. Wolf, Conditional entropies and their relation to entanglement criteria, \href{\doibase 10.1063/1.1498490}{J. Math. Phys. \textbf{43}, 4299 (2002)}.
		
		\bibitem{hayashi2003}
		T. Hayashi, T. Fujisawa, H. D. Cheong, Y. H. Jeong, and Y. Hirayama, Coherent Manipulation of Electronic States in a Double Quantum Dot, \href{\doibase 10.1103/PhysRevLett.91.226804}{Phys. Rev. Lett. \textbf{91}, 226804 (2003)}.
		
		\bibitem{gorman2005}
		J. Gorman, D. G. Hasko, and D. A. Williams, Charge-Qubit Operation of an Isolated Double Quantum Dot, \href{\doibase 10.1103/PhysRevLett.95.090502}{Phys. Rev. Lett. \textbf{95}, 090502 (2005)}.
		
		\bibitem{scandi2022}
		M. Scandi, D. Barker, S. Lehmann, K. A. Dick, V. F. Maisi, and M. Perarnau-Llobet, Minimally Dissipative Information Erasure in a Quantum Dot via Thermodynamic Length, \href{\doibase 10.1103/PhysRevLett.129.270601}{Phys. Rev. Lett. \textbf{129}, 270601 (2022)}.
		
		\bibitem{koski2014}
		J. V. Koski, V. F. Maisi, J. P. Pekola, and D. V. Averin, Experimental realization of a Szilard engine with a single electron, \href{\doibase 10.1073/pnas.1406966111}{Proc. Natl. Acad. Sci. U.S.A. \textbf{111}, 13786 (2014)}.
		
		\bibitem{fichetti1998}
		M. V. Fichetti, Theory of electron transport in small semiconductor devices using the Pauli master equation, \href{\doibase 10.1063/1.367149}{J. Appl. Phys. \textbf{83}, 270 (1998)}.
		
		\bibitem{benenti2017}
		G. Benenti, G. Casati, K. Saito, and R. S. Whitney, Fundamental aspects of steady-state conversion of heat to work at the nanoscale, \href{\doibase 10.1016/j.physrep.2017.05.008}{Phys. Rep. \textbf{694}, 1 (2017)}.
		
		\bibitem{schaller2014}
		G. Schaller, \textit{Open Quantum Systems Far from Equilibrium} (Springer, Heidelberg, 2014).
		
		\bibitem{peschel2003}
		I. Peschel, Calculation of reduced density matrices from correlation functions, \href{\doibase 10.1088/0305-4470/36/14/101}{J. Phys. A: Math. Gen. \textbf{36}, L205 (2003)}.
		
		\bibitem{eisler2012}
		V. Eisler and I. Peschel, On entanglement evolution across defects in critical chains, \href{\doibase 10.1209/0295-5075/99/20001}{Europhys. Lett. \textbf{99}, 20001 (2012)}.
		
		\bibitem{bruss2004}
		D. Bruß, G. M. D'Ariano, M. Lewenstein, C. Macchiavello, A. Sen(De), and U. Sen, Distributed Quantum Dense Coding, \href{\doibase 10.1103/PhysRevLett.93.210501}{Phys. Rev. Lett. \textbf{93}, 210501 (2004)}.
		
		\bibitem{prabhu2013}
		R. Prabhu, A. K. Pati, A. Sen(De), and U. Sen, Exclusion principle for quantum dense coding, \href{\doibase 10.1103/PhysRevA.87.052319}{Phys. Rev. A \textbf{87}, 052319 (2013)}.
		
		\bibitem{yang2019}
		D. Yang, K. Horodecki, and A. Winter, Distributed Private Randomness Distillation, \href{\doibase 10.1103/PhysRevLett.123.170501}{Phys. Rev. Lett. \textbf{123}, 170501 (2019)}.
		
		\bibitem{vempati2021}
		M. Vempati, N. Ganguly, I. Chakrabarty, and A. K. Pati, Witnessing negative conditional entropy, \href{\doibase 10.1103/PhysRevA.104.012417}{Phys. Rev. A \textbf{104}, 012417 (2021)}.
	
		\bibitem{verstraete2003}
		F. Verstraete and J. I. Cirac, Quantum Nonlocality in the Presence of Superselection Rules and Data Hiding Protocols, \href{\doibase 10.1103/PhysRevLett.91.010404}{Phys. Rev. Lett. \textbf{91}, 010404 (2003)}.
		
		\bibitem{schuch2004}
		N. Schuch, F. Verstraete, and J. I. Cirac, Quantum entanglement theory in the presence of superselection rules, \href{\doibase 10.1103/PhysRevA.70.042310}{Phys. Rev. A \textbf{70}, 042310 (2004)}.
		
		\bibitem{morris2020}
		B. Morris, B. Yadin, M. Fadel, T. Zibold, P. Treutlein, and G. Adesso, Entanglement between Identical Particles Is a Useful and Consistent Resource, \href{\doibase 10.1103/PhysRevX.10.041012}{Phys. Rev. X \textbf{10}, 041012 (2020)}.
		
	\end{thebibliography}
\end{document}